\def\FossacsPreprint{1}
\documentclass[runningheads]{llncs}
\usepackage[T1]{fontenc}
\usepackage{amsmath,amssymb}
\usepackage{stmaryrd}
\usepackage{mathpartir}
\usepackage{booktabs}
\usepackage{tikz}
\usetikzlibrary{arrows.meta}
\usepackage[hidelinks]{hyperref}
\hypersetup{pdftitle={An Explicit Ordinal Bound for System T Dialogue Trees}}
\AddToHook{env/thebibliography/begin}{\raggedright\urlstyle{same}}
\newcommand{\N}{\mathbb{N}}
\newcommand{\Tree}[1]{\mathcal{D}(#1)}
\newcommand{\eps}{\varepsilon_0}
\newcommand{\tow}[1]{\theta_{#1}}
\newcommand{\den}[1]{\llbracket #1\rrbracket}
\newcommand{\dend}[1]{\llbracket #1\rrbracket^{D}}
\newcommand{\deni}[1]{\llbracket #1\rrbracket^{\infty}}
\newcommand{\lit}[1]{\underline{#1}}
\newcommand{\caseop}{\mathsf{case}}
\newcommand{\query}{\mathsf{query}}
\newcommand{\Rec}{\mathsf{Rec}}
\newcommand{\Zero}{\mathsf{0}}
\newcommand{\Succ}{\mathsf{succ}}
\newcommand{\gen}{\mathsf{gen}}
\newcommand{\pass}[1]{P_{#1}}
\newcommand{\nclass}[1]{\mathcal{N}_{#1}}
\newcommand{\bind}{\mathbin{>\!\!>\!\!=}}
\newcommand{\ordjoin}{\mathbin{\vee}}
\newcommand{\lev}{\operatorname{lev}}

\title{An Explicit Ordinal Bound for System T Dialogue Trees}
\titlerunning{An Explicit Ordinal Bound for System T Dialogue Trees}
\ifdefined\FossacsPreprint
\author{Mingkun Xiao\inst{1, \star}
\and Yixuan Sun\inst{2,\star}}
\authorrunning{M. Xiao and Y. Sun}
\institute{Capital Normal University\\
\email{1230604016@cnu.edu.cn}
\and
Kean University\\
\email{sunyixua@kean.edu}}
\hypersetup{pdfauthor={Mingkun Xiao; Yixuan Sun}}

\hypersetup{pdfsubject={Preprint}}
\else
\author{Anonymous Author(s)}
\authorrunning{Anonymous Author(s)}
\institute{}
\hypersetup{pdfauthor={Anonymous Author(s)},
  pdfsubject={FoSSaCS 2027 working draft}}
\fi
\begin{document}
\maketitle
\begin{abstract}
Escard\'o's dialogue interpretation assigns to each closed term
$t:(\iota\to\iota)\to\iota$ of G\"odel's System~T a well-founded,
countably branching tree $D(t)$, where $\iota$ is the natural-number type.
We give a direct proof that its classical ordinal height is below $\eps$.
More precisely, we compute a natural number $K(t)\ge2$ from the type levels
occurring in the source term and prove $h(D(t))<\tow{K(t)}$,
where $\tow0=\omega$ and
$\tow{n+1}=\omega^{\tow n}$.
Our proof translates recursors into closed infinitary templates and
eliminates $\beta$-redexes by a finite sequence of passes indexed by
ordinary type level. The translation and every pass preserve the dialogue
denotation exactly. An auxiliary rank $\rho$ satisfies an additive
substitution bound; each pass sends rank $\alpha$ to at most $2^\alpha$.
Combining these estimates with a computable initial bound $\omega+m(t)$
and a dialogue-height bound $2^{\rho(N)}$ for closed ground normal forms
$N$ yields the stated tower bound. A semantics-preserving translation
transfers the result to Escard\'o's original combinatory interpretation.
We formalise the proof in Agda over classical ordinals under explicit
foundational assumptions.

\end{abstract}
\section{Introduction}
\label{sec:problem}

A higher-order program can interact with its input by calling a function
supplied as an argument. For a functional $F:(\N\to\N)\to\N$, the answer
returned by its input $\alpha$ may determine the next query. A dialogue tree
records these interactions: internal nodes specify queries, branches
represent possible answers, and leaves record final outputs.
The numerical result alone does not determine this tree.

For well-founded trees, each path is finite, but their lengths need not
have a common finite bound. Consider a tree that first queries $\alpha(0)$
and, on receiving $n$, performs $n$ further queries before returning zero.
Giving leaves height zero and each node the supremum of the successors
of its child heights gives this tree height $\sup_n(n+1)=\omega$;
see~\eqref{eq:tree-height}. We ask which bounds on ordinal height follow
from restrictions on the program that generates the tree.

G\"odel's System~T~\cite{goedel1958} extends the simply typed
$\lambda$-calculus over the natural-number type $\iota$ with primitive
recursion at all finite types.
Escard\'o's dialogue interpretation assigns a specified
well-founded tree $D(t)$ to every closed term
$t:(\iota\to\iota)\to\iota$~\cite{escardo2013}.
Following this tree with an input $\alpha$ returns the value of $t$ in
the usual functional interpretation. This representation yields continuity
information; subsequent work also internalises the tree representation
in System~T~\cite{escardo2025}. Kawai~\cite{kawai2019} internalises the
interpretation in finite-type intuitionistic arithmetic using neighbourhood functions,
which determine outputs from finite input prefixes. Xu~\cite{xu2020}
gives a framework of translations and logical relations for continuity
and related properties, formalised in Agda. We study the ordinal height of
the tree assigned by the fixed dialogue interpretation.
Example~\ref{ex:omega-height} gives a System~T term whose specified
dialogue also has height $\omega$.

The ordinal $\eps$, familiar from ordinal analyses of finite-type
recursion~\cite{tait1965,howard1980}, is a natural candidate for a common
bound. With
\[
 \tow0=\omega,\qquad \tow{n+1}=\omega^{\tow n},
 \qquad \eps=\sup_n\tow n,
\]
it is the supremum of the finite $\omega$-exponential towers.

Escardó’s 2013 paper states the $\eps$ height bound as a corollary
without giving the corresponding ordinal argument~\cite{escardo2013}. More recently,
Escardó et al. explicitly formulated the $\eps$ bound as a natural conjecture~\cite{escardo2025}.
We give a direct proof of this bound for the specified dialogue interpretation,
together with an explicit finite-tower bound computable from the source term.

\begin{theorem}[Source-dependent height bound]
\label{thm:main}
There is a syntactically computable function $K$ from closed System~T terms
$t:(\iota\to\iota)\to\iota$ to natural numbers at least two such that
\[
 h(D(t))<\tow{K(t)}<\eps.
\]
\end{theorem}

The parameter $K(t)$ is obtained from the levels of the types occurring
in the source term and its instantiated recursors. For each fixed type-level ceiling,
the theorem gives a uniform ordinal bound on the resulting dialogue trees.
The proof gives a finer bound by iterating base-two ordinal exponentiation
$K(t)+1$ times from a computable ordinal $\omega+m(t)$;
see~\eqref{eq:fine-height-bound}.

The specified interpretation matters to the proof. The introductory tree
returns zero on every input and has height $\omega$,
whereas a single leaf returning zero has height zero. Thus correctness of a
transformation at the level of numerical results is insufficient.
The dialogue interpretation also fails to validate unrestricted source
recursor-conversion equations on effectful arguments~\cite{escardo2025}.
Our construction maintains exact dialogue-tree equality
throughout its transformations.

Our analysis follows the infinitary-term approach of Tait~\cite{tait1965}
and Martin-L\"of~\cite{martinlof1972}.
Howard~\cite[\S3]{howard1980} establishes $\eps$-bounds for computation trees and
trees of unsecured sequences associated with finite-type terms.
Wilken and Weiermann~\cite{wilken2012} classify reduction lengths for
System~T and its fragments. Connecting these analyses to $D(t)$
requires a comparison of tree representations. Our translation and
normalisation passes preserve $D(t)$ exactly.

We translate the source into well-founded, countably branching terms.
Each recursor becomes a closed template of finite iterates. Keeping its
parameters bound gives a uniform initial rank estimate; actual arguments
enter through application. Passes indexed by ordinary type level then
eliminate $\beta$-redexes. Strict domain descent controls redexes created
by substitution. Codomain levels do not increase: an abstraction exposed
at the current level is contracted by the enclosing application in the
same pass. The finite source supplies a common type ceiling for all
branches, and hence a finite number of passes.

The ordinal rank $\rho$ satisfies
$\rho(M[\vartheta])\le\beta+\rho(M)$ for substitution by terms of rank
at most $\beta$. One pass sends rank $\alpha$ to at most $2^\alpha$,
and the dialogue height of a closed ground normal form $N$ is at most
$2^{\rho(N)}$. These estimates yield the bounds above.

The proof is formalised in Agda~\cite{norell2007} using
TypeTopology~\cite{typetopology}, with the classical ordinal model's
assumptions explicit. Section~\ref{sec:dialogue} fixes the dialogue
interpretation. Sections~\ref{sec:translation}--\ref{sec:rank} develop the
translation, normalisation, and bounds; Section~\ref{sec:combinatory}
transfers the result to the original combinatory interpretation.
Section~\ref{sec:mechanisation} describes the mechanisation.

\section{System T and Dialogue Semantics}
\label{sec:dialogue}

We fix the $\lambda$-calculus presentation of System~T and its dialogue
interpretation~\cite{escardo2013,escardo2025}. We use classical ordinal
arithmetic under the assumptions stated in Section~\ref{sec:mechanisation}.

\subsection{System T}

The simple types are generated by
\[
 \sigma,\tau ::= \iota \mid \sigma\to\tau,
\]
where $\iota$ is the type of natural numbers. A context $\Gamma$ is a finite
list of distinct variables with assigned types, and $\Gamma\vdash t:\sigma$
means that $t$ has type $\sigma$ in that context. Terms are formed from
variables, abstraction $\lambda x^\sigma.t$, application $t\,u$, and the
constants
\begin{equation}\label{eq:t-constants}
 \begin{aligned}
 \Zero&:\iota,\qquad \Succ:\iota\to\iota,\\
 \Rec_\sigma&:(\iota\to\sigma\to\sigma)\to\sigma\to\iota\to\sigma.
 \end{aligned}
\end{equation}
Arrows associate to the right and application to the left. We use the
usual typing rules, identify terms up to renaming of bound variables,
and write $t[u/x]$ for capture-avoiding substitution. A term is closed
if it has no free variables; we consider such terms in the empty context.

The usual interpretation $\den{\cdot}$ takes $\iota$ to $\N$, arrows to
function spaces, and the constants to zero, successor, and primitive
recursion. Here $\iota$ is a source type and $\N$ the natural numbers
of the metatheory.

\subsection{Dialogue Trees}

For a set $X$, the type $\Tree X$ of dialogue trees is inductively generated by
\[
 \eta(x)\quad(x\in X),\qquad
 \beta_i(\varphi)\quad(i\in\N,\ \varphi:\N\to\Tree X).
\]
A leaf $\eta(x)$ returns $x$. A node $\beta_i(\varphi)$ queries position
$i$ and continues with $\varphi(n)$ when the answer is $n$.
Inductive generation makes the trees well-founded and supports structural induction;
it does not impose a common finite bound on branch depth.

Given an input function $\alpha:\N\to\N$, evaluation returns an element
of $X$ by
\begin{equation}\label{eq:tree-evaluation}
 \begin{aligned}
 \operatorname{eval}(\eta(x),\alpha)&=x,\\
 \operatorname{eval}(\beta_i(\varphi),\alpha)
   &=\operatorname{eval}(\varphi(\alpha(i)),\alpha).
 \end{aligned}
\end{equation}
For $d\in\Tree\N$, the height $h(d)$ is the classical ordinal defined by
\begin{equation}\label{eq:tree-height}
 h(\eta(n))=0,\qquad
 h(\beta_i(\varphi))=\sup_{n\in\N}\bigl(h(\varphi(n))+1\bigr).
\end{equation}
It ignores the leaf labels and query positions. The supremum includes
every branch of the inductive tree; evaluation follows only answers
consistent with the chosen $\alpha$.

For $d\in\Tree X$ and $g:X\to\Tree Y$, grafting replaces each leaf
$\eta(x)$ of $d$ by $g(x)$. We write it as $d\bind g$ and define
\begin{equation}\label{eq:tree-bind}
 \begin{aligned}
 \eta(x)\bind g&=g(x),\\
 \beta_i(\varphi)\bind g
   &=\beta_i\bigl(n\mapsto\varphi(n)\bind g\bigr).
 \end{aligned}
\end{equation}
In particular, for $f:X\to Y$, leaf relabelling is
$\operatorname{map}(f,d)=d\bind(x\mapsto\eta(f(x)))$.
Tree equality throughout means equality in this inductive type;
agreement of evaluations alone does not identify trees.

\subsection{Dialogue Interpretation}

The dialogue values at each source type are
\begin{equation}\label{eq:dialogue-values}
 V_\iota=\Tree\N,\qquad
 V_{\sigma\to\tau}=V_\sigma\to V_\tau.
\end{equation}
Only ground values are trees. A valuation $\gamma$ for $\Gamma$ assigns
$\gamma(x)\in V_\sigma$ to each declaration $x:\sigma$ in $\Gamma$.
For fresh $x:\sigma$ and $a\in V_\sigma$, write $\gamma[x\mapsto a]$
for its extension assigning $a$ to $x$.

For $\Gamma\vdash t:\sigma$, its dialogue interpretation
$\dend t_\gamma\in V_\sigma$ is defined by recursion on $t$.
The variable, abstraction, and application clauses are
\begin{equation}\label{eq:dialogue-lambda}
 \begin{aligned}
 \dend x_\gamma&=\gamma(x),\\
 \dend{\lambda x^\sigma.t}_\gamma(a)
   &=\dend t_{\gamma[x\mapsto a]},\qquad a\in V_\sigma,\\
 \dend{t\,u}_\gamma&=\dend t_\gamma(\dend u_\gamma).
 \end{aligned}
\end{equation}
We suppress the valuation on constants and closed terms. For $d\in\Tree\N$,
zero and successor are interpreted by
\begin{equation}\label{eq:dialogue-zero-succ}
 \dend\Zero=\eta(0),\qquad
 \dend\Succ(d)=\operatorname{map}(n\mapsto n+1,d).
\end{equation}

To interpret recursion at an arbitrary result type, extend grafting to
dialogue values. For $d\in\Tree X$ and $g:X\to V_\sigma$, define
$d\bind_\sigma g\in V_\sigma$ by induction on $\sigma$:
\begin{equation}\label{eq:typed-bind}
 \begin{aligned}
 d\bind_\iota g&=d\bind g,\\
 (d\bind_{\sigma\to\tau}g)(a)
   &=d\bind_\tau(x\mapsto g(x)(a)),\qquad a\in V_\sigma.
 \end{aligned}
\end{equation}
This generalised Kleisli extension~\cite[\S3]{escardo2013} acts pointwise
at function types and uses~\eqref{eq:tree-bind} at ground type.

Let $f:\Tree\N\to V_\sigma\to V_\sigma$ and $a\in V_\sigma$.
Ordinary recursion on the natural number $n$ defines $v_n\in V_\sigma$ by
\[
 v_0=a,\qquad v_{n+1}=f(\eta(n))(v_n).
\]
For $d\in\Tree\N$, set
\begin{equation}\label{eq:dialogue-rec}
 \dend{\Rec_\sigma}(f)(a)(d)=d\bind_\sigma(n\mapsto v_n).
\end{equation}
The index is passed to $f$ as the leaf $\eta(n)$. At ground result type,
each leaf labelled $n$ in $d$ is replaced by $v_n$; at function types,
grafting acts pointwise.

The generic input $\gen:\Tree\N\to\Tree\N$ is
\begin{equation}\label{eq:generic-input}
 \gen(d)=d\bind\bigl(i\mapsto\beta_i(n\mapsto\eta(n))\bigr).
\end{equation}
It replaces each leaf $\eta(i)$ by a query at position $i$ returning the
answer. For closed $t:(\iota\to\iota)\to\iota$, its specified tree is
\begin{equation}\label{eq:specified-dialogue}
 D(t)=\dend t(\gen)\in\Tree\N.
\end{equation}
The established correctness theorem~\cite{escardo2013,escardo2025} states that, for every
$\alpha:\N\to\N$,
\[
 \operatorname{eval}(D(t),\alpha)=\den t(\alpha).
\]
\begin{example}[A dialogue of height $\omega$]\label{ex:omega-height}
Consider the closed term
\[
 t_\omega=\lambda\alpha^{\iota\to\iota}.\,
 \Rec_\iota(\lambda i^\iota.\lambda r^\iota.\alpha(r))
 \,\Zero\,(\alpha(\Zero)).
\]
It reads $\alpha(0)=n$ and iterates $\alpha$ $n$ times from $0$.
The dialogue step applies $\gen$ to the previous result, ignoring the
index. Since $\gen(\eta(0))=\beta_0(n\mapsto\eta(n))$,
\eqref{eq:dialogue-rec} gives
\[
 v_0=\eta(0),\quad v_{n+1}=\gen(v_n),\quad
 D(t_\omega)=\beta_0(n\mapsto v_n).
\]
Each application of $\gen$ adds one query at every leaf, so induction
on $n$ shows that every root-to-leaf path in $v_n$ has exactly $n$ query nodes.
Consequently,
\[
 h(v_n)=n,\qquad h(D(t_\omega))=\sup_n(n+1)=\omega.
\]
A finite term with a single recursor at ground result type thus produces
a dialogue of infinite ordinal height.
\end{example}

The interpretation is sensitive to the syntax of recursion. For source
terms $s:\iota\to\sigma\to\sigma$, $a:\sigma$, and $u:\iota$, the usual
recursor reduction
\[
 \Rec_\sigma\,s\,a\,(\Succ\,u)
   \longrightarrow s\,u\,(\Rec_\sigma\,s\,a\,u)
\]
need not preserve dialogue trees: the right-hand side can repeat the
queries made by $u$~\cite{escardo2025}. Our later transformations therefore
use the recursor interpretation~\eqref{eq:dialogue-rec} and
$\beta$-contraction in the translated language, with tree equality as
their semantic invariant.

\section{Infinitary Translation}
\label{sec:translation}

We expand recursors into finite iterates, following
Tait~\cite{tait1965} and Martin-L\"of~\cite{martinlof1972}, while keeping
abstraction and application explicit. Add $\Omega:\iota\to\iota$ to the
source with $\dend\Omega=\gen$, so that $\dend{t\,\Omega}=D(t)$.

In Figure~\ref{fig:translation-overview}, source and auxiliary terms
share the object-language types $\sigma$; $\N$, $\Tree\N$, and $V_\sigma$
are value spaces in the metatheory. All arrows are metatheoretic maps.
The templates and finite iterates below are auxiliary terms.

\begin{figure}[ht]
\centering
\begin{tikzpicture}[
  font=\small,
  value/.style={align=center,inner sep=3pt},
  syntax note/.style={align=center,font=\scriptsize,anchor=south,inner sep=0pt},
  map/.style={-{Stealth[length=4pt,width=3pt]},line width=.45pt},
  edge label/.style={font=\footnotesize,inner sep=2pt,fill=white}
]
\node[value] (source) at (0,0) {$t$};
\node[value] (extended) at (2.6,0) {$t\,\Omega$};
\node[value] (translated) at (5.2,0) {$M_t$};
\node[value] (normal) at (7.8,0) {$N_t$};
\node[syntax note,yshift=7pt] at (source.north)
  {System T\\[-1pt]$(\iota\to\iota)\to\iota$};
\node[syntax note,yshift=7pt] at (extended.north)
  {source with $\Omega$\\[-1pt]$\iota$};
\node[syntax note,yshift=7pt] at (translated.north)
  {infinitary\\[-1pt]$\iota$};
\node[syntax note,yshift=7pt] at (normal.north)
  {infinitary, $\beta$-normal\\[-1pt]$\iota$};

\node[value] (dialogue-value) at (0,-1.9)
  {$\dend t$\\[-1pt]$\in V_{(\iota\to\iota)\to\iota}$};
\node[value] (tree) at (5.2,-1.9)
  {$D(t)\in\Tree\N$};
\node[value] (functional) at (0,-3.45)
  {$\den t$\\[-1pt]$\in(\N\to\N)\to\N$};
\node[value] (result) at (5.2,-3.45)
  {$\den t(\alpha)\in\N$};

\draw[map] (source.east) -- node[edge label,above=2pt]
  {$(-)\,\Omega$} (extended.west);
\draw[map] (extended.east) -- node[edge label,above=2pt]
  {$\widehat{(-)}$} (translated.west);
\draw[map] (translated.east) -- node[edge label,above=2pt]
  {$P_1\circ\cdots\circ P_{K(t)}$} (normal.west);

\draw[map] (source.south) -- node[edge label,right]
  {$\dend{-}$} (dialogue-value.north);
\draw[map] (extended.south) -- node[edge label,left]
  {$\dend{-}$} (tree.north west);
\draw[map] (translated.south) -- node[edge label,right]
  {$\deni{-}$} (tree.north);
\draw[map] (normal.south) -- node[edge label,right]
  {$\deni{-}$} (tree.north east);
\draw[map] (dialogue-value.east) -- node[edge label,above]
  {$(-)(\gen)$} (tree.west);

\draw[map] (source.west) .. controls (-2,-.75) and (-2,-2.7) ..
  node[edge label,left,pos=.52] {$\den{-}$} (functional.west);
\draw[map] (functional.east) -- node[edge label,above]
  {$(-)(\alpha)$} (result.west);
\draw[map] (tree.south) -- node[edge label,right]
  {$\operatorname{eval}(-,\alpha)$} (result.north);
\end{tikzpicture}
\caption{Syntax and semantics for a fixed closed term $t$ and input
$\alpha:\N\to\N$. The upper row contains terms; the lower rows contain
semantic values. The paths to $D(t)$ agree by
Theorem~\ref{thm:translation} and semantic preservation of normalisation.
The paths to $\den t(\alpha)$ agree by evaluation correctness
(Section~\ref{sec:dialogue}).}
\label{fig:translation-overview}
\end{figure}
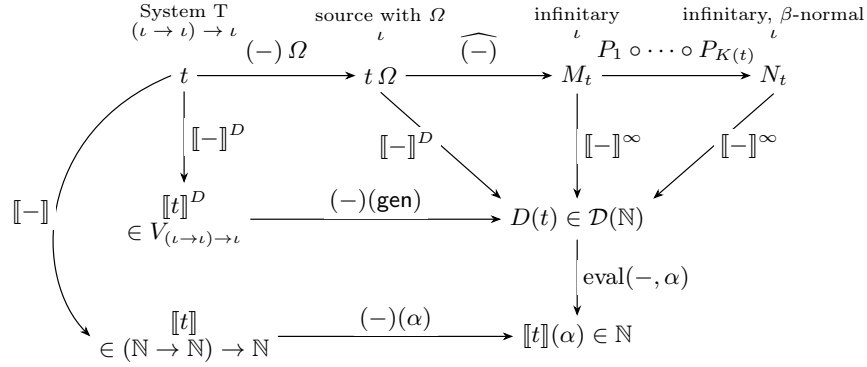

\subsection{Infinitary Terms}

The auxiliary terms $M,N$ use the same types and contexts as the source
language. They are inductively generated by variables, abstraction, and
application with their usual typing rules, together with
\[
 \lit n,\qquad \query_i(n\mapsto M_n),\qquad
 \caseop(M;\,n\mapsto M_n).
\]
Here $i,n\in\N$, and $n\mapsto M_n$ denotes a family of terms indexed
by the natural numbers of the metatheory.
Literals have type $\Gamma\vdash\lit n:\iota$ in any context.
The two branching constructs are typed by
\begin{mathpar}
 \inferrule{\Gamma\vdash M_n:\iota\quad(n\in\N)}
   {\Gamma\vdash\query_i(n\mapsto M_n):\iota}
 \and
 \inferrule{\Gamma\vdash M:\iota \\ \Gamma\vdash M_n:\iota\quad(n\in\N)}
   {\Gamma\vdash\caseop(M;\,n\mapsto M_n):\iota}.
\end{mathpar}
Both constructs have ground result type and all branches share a context.
Selection at function types is defined below using abstraction.

These well-founded syntax trees may have countably many children with
unbounded finite depths; structural induction has one hypothesis per child.
We identify bound variables up to renaming. Capture-avoiding substitution
extends to the selector and every branch.

For $\Gamma\vdash M:\sigma$, define $\deni M_\gamma\in V_\sigma$
by structural recursion, using the value spaces and valuations of
Section~\ref{sec:dialogue}. Variables, abstraction, and application are
interpreted as in~\eqref{eq:dialogue-lambda}. The new clauses are
\begin{equation}\label{eq:infinitary-semantics}
 \begin{aligned}
 \deni{\lit n}_\gamma&=\eta(n),\\
 \deni{\query_i(n\mapsto M_n)}_\gamma
   &=\beta_i(n\mapsto\deni{M_n}_\gamma),\\
 \deni{\caseop(M;\,n\mapsto M_n)}_\gamma
   &=\deni M_\gamma\bind(n\mapsto\deni{M_n}_\gamma).
 \end{aligned}
\end{equation}
Selection retains the selector's queries and grafts branch denotations
onto its leaves. These structural definitions do not presuppose normalisation.

\subsection{Translation}

For $\Gamma\vdash M:\iota$ and $\Gamma\vdash M_n:\sigma$ for every
$n\in\N$, define a derived term
$\Gamma\vdash\caseop_\sigma(M;\,n\mapsto M_n):\sigma$
by induction on $\sigma$:
\begin{equation}\label{eq:derived-selection}
 \begin{aligned}
 \caseop_\iota(M;\,n\mapsto M_n)
   &=\caseop(M;\,n\mapsto M_n),\\
 \caseop_{\sigma\to\tau}(M;\,n\mapsto M_n)
   &=\lambda x^\sigma.\caseop_\tau(M;\,n\mapsto M_n\,x).
 \end{aligned}
\end{equation}
The variable $x$ is fresh for $\Gamma$; the selector and branches on the
right are regarded as terms in the extended context. At a function type,
selection takes an argument and applies each candidate function to it
before selecting at the result type.

\begin{lemma}[Derived selection]\label{lem:derived-selection}
For the terms above and every valuation $\gamma$ for $\Gamma$,
\[
 \deni{\caseop_\sigma(M;\,n\mapsto M_n)}_\gamma
 =\deni M_\gamma\bind_\sigma(n\mapsto\deni{M_n}_\gamma).
\]
\end{lemma}
\begin{proof}
By induction on $\sigma$. The ground case
is~\eqref{eq:infinitary-semantics}. For $\sigma\to\tau$, put
$d=\deni M_\gamma$ and $g(n)=\deni{M_n}_\gamma$.
Apply the left-hand side to any $a\in V_\sigma$.
By~\eqref{eq:derived-selection}, freshness of $x$, and the induction
hypothesis at $\tau$, the result is
$d\bind_\tau(n\mapsto g(n)(a))$.
Equation~\eqref{eq:typed-bind} identifies this with
$(d\bind_{\sigma\to\tau}g)(a)$, giving equality of the functions.
\qed
\end{proof}

For terms $f:\iota\to\sigma\to\sigma$ and $a:\sigma$ in a common
context, ordinary recursion on $n\in\N$ defines
\begin{equation}\label{eq:finite-iterates}
 I_0(f,a)=a,\qquad I_{n+1}(f,a)=f\,\lit n\,(I_n(f,a)).
\end{equation}
For example, $I_2(f,a)=f\,\lit1\,(f\,\lit0\,a)$.
The closed recursor template is
\begin{equation}\label{eq:recursor-template}
 R^\infty_\sigma
 =\lambda f^{\iota\to\sigma\to\sigma}\,a^\sigma\,z^\iota.
   \caseop_\sigma(z;\,n\mapsto I_n(f,a)),
\end{equation}
of type $(\iota\to\sigma\to\sigma)\to\sigma\to\iota\to\sigma$.
Its $n$th branch represents $n$ iterations of the step function, with
the indices $0,\ldots,n-1$ passed as literals.

The translation $u\mapsto\widehat u$ is defined by recursion on the
finite source syntax. Its constant clauses are
\begin{equation}\label{eq:translation-constants}
 \begin{aligned}
 \widehat{\Zero}&=\lit0,\qquad
   \widehat{\Rec_\sigma}=R^\infty_\sigma,\\
 \widehat{\Succ}&=\lambda z^\iota.\caseop(z;\,n\mapsto\lit{n+1}),\\
 \widehat\Omega&=\lambda z^\iota.\caseop
   (z;\,i\mapsto\query_i(n\mapsto\lit n)).
 \end{aligned}
\end{equation}
The successor template relabels leaves, and the oracle template grafts
a query at each returned position. The remaining clauses preserve the
$\lambda$-calculus structure:
\[
 \widehat x=x,\qquad
 \widehat{\lambda x^\sigma.u}=\lambda x^\sigma.\widehat u,\qquad
 \widehat{u\,v}=\widehat u\,\widehat v.
\]
In particular,
\[
 \widehat{\Rec_\sigma\,s\,a\,u}
   =R^\infty_\sigma\,\widehat s\,\widehat a\,\widehat u.
\]
Each $I_n(f,a)$ contains only variables, literals, and applications.
Actual step and base arguments remain outside the closed template until
$\beta$-contraction. Section~\ref{sec:rank} bounds its initial rank
independently of those arguments, then combines the estimates through
the finite source syntax.

\subsection{Semantic Preservation}

\begin{theorem}[Semantic preservation]\label{thm:translation}
If $\Gamma\vdash u:\sigma$ in the source language extended by $\Omega$,
then $\Gamma\vdash\widehat u:\sigma$ in the auxiliary language, and
for every valuation $\gamma$ for $\Gamma$,
\[
 \deni{\widehat u}_\gamma=\dend u_\gamma.
\]
\end{theorem}
\begin{proof}
The template types and the common typing rules give type preservation
by induction on $u$. For semantic preservation, first consider the
recursor template. Fix $F\in V_{\iota\to\sigma\to\sigma}$,
$A\in V_\sigma$, and $d\in\Tree\N$.
Let $v_0=A$ and $v_{n+1}=F(\eta(n))(v_n)$, as in the source recursor
interpretation. In a valuation $\delta$ assigning $f,a,z$ the values
$F,A,d$, respectively, induction on $n$ in~\eqref{eq:finite-iterates}
gives
\[
 \deni{I_n(f,a)}_\delta=v_n.
\]
The base case is the interpretation of $a$; the successor case applies
$F$ to $\eta(n)$ and to the induction hypothesis.
Lemma~\ref{lem:derived-selection} now yields
\[
 \deni{R^\infty_\sigma}(F)(A)(d)
   =d\bind_\sigma(n\mapsto v_n)
   =\dend{\Rec_\sigma}(F)(A)(d).
\]
Thus the two recursor interpretations agree at every argument.
For zero, successor, and $\Omega$, the clauses
of~\eqref{eq:translation-constants} give precisely
\eqref{eq:dialogue-zero-succ} and~\eqref{eq:generic-input}.
Finally, structural induction on $u$ completes the proof: variables
use the same valuation, application preserves equality, and abstraction
uses the induction hypothesis in each extended valuation.
\qed
\end{proof}

For a closed source term $t:(\iota\to\iota)\to\iota$, set
\begin{equation}\label{eq:translated-dialogue}
 M_t=\widehat{t\,\Omega},
 \qquad \deni{M_t}=\dend{t\,\Omega}=D(t).
\end{equation}
The term $M_t$ is closed and has ground type. The equality is in the
inductive type of dialogue trees and retains query positions and every
answer branch. Section~\ref{sec:normalisation} transforms $M_t$ to a
$\beta$-normal term while preserving this denotation.

\section{Normalisation}
\label{sec:normalisation}

We transform $M_t$ into a closed ground $\beta$-normal term $N_t$ with
$\deni{N_t}=D(t)$. The finite source supplies a common type bound;
structural passes eliminate one redex degree at a time.
Section~\ref{sec:rank} estimates the resulting rank and dialogue height.

\subsection{Type Levels and the Source Ceiling}

We use ordinary type level, as in Martin-L\"of~\cite[\S4]{martinlof1972}:
\begin{equation}\label{eq:type-degree}
 \lev(\iota)=0,\qquad
 \lev(\sigma\to\tau)=\max\{\lev(\sigma)+1,\lev(\tau)\}.
\end{equation}
Thus $\lev(\sigma)<\lev(\sigma\to\tau)$ and
$\lev(\tau)\le\lev(\sigma\to\tau)$. The degree of a redex
$(\lambda x^\sigma.U)\,V$, with $U:\tau$, is
$\lev(\sigma\to\tau)$. For each context, type, and $r\in\N$, let $\nclass r$
be the class of auxiliary terms whose redexes all have degree strictly below $r$,
including those under abstractions and in every branch. We suppress
the context and type in this notation. Membership is inherited by
subterms. Since every arrow type has positive level, $M\in\nclass1$
means exactly that $M$ is $\beta$-normal.

We compute the ceiling on the finite extended source syntax. Define
$B(u)\in\N$ as follows, with superscripts indicating the types of
subterms:
\begin{equation}\label{eq:source-ceiling}
 \begin{aligned}
  B(\Zero)&=0,\qquad B(\Succ)=B(\Omega)=1,\\
  B(x^\sigma)&=\lev(\sigma),\\
  B(\Rec_\sigma)
    &=\lev((\iota\to\sigma\to\sigma)\to\sigma\to\iota\to\sigma),\\
  B(\lambda x^\sigma.u^\tau)
    &=\max\{\lev(\sigma\to\tau),B(u)\},\\
  B(u^{\sigma\to\tau}\,v^\sigma)
    &=\max\{\lev(\tau),B(u),B(v)\}.
 \end{aligned}
\end{equation}
For the functional $t$ of interest, put
\begin{equation}\label{eq:source-K}
 K(t)=B(t\,\Omega).
\end{equation}
By definition, $K(t)\ge B(t)\ge\lev((\iota\to\iota)\to\iota)=2$.

\begin{lemma}[Uniform type ceiling]\label{lem:uniform-ceiling}
For every finite extended source term $u$, the result type of every
subterm of $\widehat u$ has level at most $B(u)$. In particular,
$M_t\in\nclass{K(t)+1}$.
\end{lemma}
\begin{proof}
First consider the derived selection
$\caseop_\sigma(M;\,n\mapsto M_n)$. If the selector and all branches
have their subterm types bounded by $r$, and $\lev(\sigma)\le r$,
then so does the derived selection. This follows by induction on
$\sigma$ in~\eqref{eq:derived-selection}: at an arrow type
$\sigma_1\to\tau$, the new variable has type $\sigma_1$, the new applications
have result type $\tau$, and the outer abstraction has type
$\sigma_1\to\tau$. Both component levels are at most $\lev(\sigma_1\to\tau)$,
and weakening preserves all types.

For $R^\infty_\sigma$, take $r$ to be the level of its recursor
type. The types of $f,a,z$, and the result types of the three outer
abstractions, are subformulas of that type. Induction on $n$ shows that
every subterm of $I_n(f,a)$ has result type of level at most this same $r$: the
successor step adds a literal and applications with result types
$\sigma\to\sigma$ and $\sigma$. The derived-selection bound therefore
applies to the entire branch family. The successor and oracle templates
use only $\iota$ and $\iota\to\iota$, and zero is a ground literal.

Structural induction on the finite source term now proves the claimed
bound, using the maxima in~\eqref{eq:source-ceiling} at abstractions
and applications. Taking $u=t\,\Omega$, a redex of degree at least
$K(t)+1$ would have an abstraction subterm whose type violates the bound.
\qed
\end{proof}

\subsection{Eliminating One Degree}

For $r\in\N$, define $\pass r$ by structural recursion. It fixes variables
and literals and commutes with abstraction and the branching constructs:
\begin{equation}\label{eq:pass-structural}
 \begin{aligned}
 \pass r(\lambda x^\sigma.M)&=\lambda x^\sigma.\pass r M,\\
 \pass r(\query_i(n\mapsto M_n))
   &=\query_i(n\mapsto\pass r M_n),\\
 \pass r(\caseop(M;\,n\mapsto M_n))
   &=\caseop(\pass r M;\,n\mapsto\pass r M_n).
 \end{aligned}
\end{equation}
For an application $F\,V$, where $F:\sigma\to\tau$, first set
$F'=\pass r F$ and $V'=\pass r V$, and then put
\begin{equation}\label{eq:pass-application}
 \pass r(F\,V)=
 \begin{cases}
  U'[V'/x],&F'=\lambda x^\sigma.U'
             \text{ and }\lev(\sigma\to\tau)=r,\\
  F'\,V',&\text{otherwise}.
 \end{cases}
\end{equation}
All recursive calls of $\pass r$ are on original immediate subterms.
The contraction clause returns the substituted term directly;
substitution is itself structural. These operations preserve context
and type and act on every child of a branching node.

Substitution can create a redex by replacing a variable in function
position with an abstraction. We first state the condition that controls
this possibility. A typed simultaneous substitution $\vartheta$ from
$\Gamma$ to $\Delta$ assigns a term $\Delta\vdash\vartheta(x):\sigma$
to each declaration $x:\sigma$ of $\Gamma$; its action on $M$ is written
$M[\vartheta]$.

\begin{lemma}[Substitution and degree]\label{lem:substitution-degree}
Let $M\in\nclass r$, and suppose every image of $\vartheta$ belongs to $\nclass r$.
Assume that each image whose root is an abstraction has type of level
strictly below $r$. Then $M[\vartheta]\in\nclass r$.
\end{lemma}
\begin{proof}
By structural induction on $M$. The variable case uses the hypothesis
on the corresponding image. Under an abstraction, lift the substitution
by leaving the fresh bound variable fixed and weakening the old images.
The fresh variable has no redex and is not an abstraction, regardless
of its type. Weakening changes neither types nor constructors, so both
hypotheses on the old images are preserved.

At an application $F\,V$, the induction hypotheses control the redexes
inside $F[\vartheta]$ and $V[\vartheta]$. If $F[\vartheta]$ has an
abstraction at its root, then, by typing, $F$ was either an abstraction or a variable:
substitution preserves an application root. In the first case, the degree
of the root redex is already below $r$ by $F\,V\in\nclass r$; in the
second, the root condition on $\vartheta$ gives this bound. Literals
have no redexes, and the query and selection cases use the induction
hypothesis for the selector, where present, and every branch.
\qed
\end{proof}

In particular, if $\Gamma,x:\sigma\vdash U:\tau$ and
$\Gamma\vdash V:\sigma$, the lemma gives
\begin{equation}\label{eq:single-substitution-degree}
 \lev(\sigma)<r,\quad U,V\in\nclass r
 \quad\Longrightarrow\quad U[V/x]\in\nclass r.
\end{equation}
Indeed, the image of $x$ has type $\sigma$, and all other variables are
left fixed. The root condition permits variables of arbitrarily high
type to remain variables, which is what makes lifting under binders
possible.

The second source of new redexes is a contraction in function position
that exposes an abstraction to an enclosing application. It is controlled
by the codomain of the contracted arrow.

\begin{lemma}[Exposed abstractions]\label{lem:exposed-abstraction}
Let $\Gamma\vdash F:\tau$. If $F$ does not have an abstraction at its
root but $\pass r F$ does, then $\lev(\tau)\le r$.
\end{lemma}
\begin{proof}
Variables, literals, queries, and selections retain their root
constructors. Hence $F$ is an application, and the first clause
of~\eqref{eq:pass-application} was used at its root. The contracted
abstraction has some type $\sigma\to\tau$ of degree $r$.
Equation~\eqref{eq:type-degree} gives $\lev(\tau)\le r$.
\qed
\end{proof}

\begin{theorem}[One-pass degree elimination]\label{thm:lowering}
For every $r\ge1$ and every well-founded infinitary term $M$,
\[
 M\in\nclass{r+1}\quad\Longrightarrow\quad\pass r M\in\nclass r.
\]
\end{theorem}
\begin{proof}
Proceed by structural induction on $M$. Only an application $F\,V$
requires a root check. Put $F'=\pass r F$ and $V'=\pass r V$;
both belong to $\nclass r$ by induction. If $F'$ is not an abstraction,
rebuilding the application creates no root redex.

Otherwise write $F'=\lambda x^\sigma.U'$ of type $\sigma\to\tau$.
Its degree is at most $r$. When $F$ was already an abstraction, this
follows from the input condition $F\,V\in\nclass{r+1}$; otherwise
Lemma~\ref{lem:exposed-abstraction} gives the same non-strict bound.
If the degree is below $r$, the application is retained and satisfies
the required condition. If it equals $r$, the pass returns $U'[V'/x]$,
including when the abstraction was exposed by processing $F$ in this pass.
Strict domain descent gives $\lev(\sigma)<r$. Both $U'$ and $V'$
belong to $\nclass r$, so~\eqref{eq:single-substitution-degree} applies.

All other constructors use the induction hypotheses on their children,
including every member of a countable branch family. The induction is
on the whole well-founded term and does not require a common finite
bound on branch depth.
\qed
\end{proof}

\begin{lemma}[Semantic preservation of a pass]\label{lem:pass-sound}
For every $r$, term $M$, and valuation $\gamma$ for its context,
\[
 \deni{\pass r M}_\gamma=\deni M_\gamma.
\]
\end{lemma}
\begin{proof}
For a substitution $\vartheta$ from $\Gamma$ to $\Delta$ and a valuation
$\gamma$ for $\Delta$, structural induction on $M$ gives
\[
 \deni{M[\vartheta]}_\gamma=\deni M_\delta,
 \qquad \delta(x)=\deni{\vartheta(x)}_\gamma.
\]
Under a binder, extend both valuations by the same value; weakening
leaves the denotations of the old substitution images unchanged.
At query and selection nodes, apply the induction hypothesis to every branch.
In particular, substitution into an abstraction body has the same
interpretation as the corresponding $\beta$-redex. Now induct on $M$:
rebuilding a constructor preserves the denotation by the induction
hypotheses, and a root contraction additionally uses this $\beta$
equation. Equality under abstractions and of branch functions uses
function extensionality. No degree bound on $M$ is needed here.
\qed
\end{proof}

\subsection{Finite-Pass Normalisation}

Define the sequence of passes by recursion on $r$:
\begin{equation}\label{eq:passes}
 Q_0(M)=M,\qquad Q_{r+1}(M)=Q_r(\pass{r+1}M),
 \qquad N_t=Q_{K(t)}(M_t).
\end{equation}
Thus $Q_r$ executes the passes with degrees $r,\ldots,1$, in that order.

\begin{theorem}[Finite-pass normalisation]\label{thm:normalisation}
For every closed System~T term $t:(\iota\to\iota)\to\iota$, the term
$N_t$ is closed, ground, and $\beta$-normal, and $\deni{N_t}=D(t)$.
\end{theorem}
\begin{proof}
Induction on $r$, using Theorem~\ref{thm:lowering}, gives
\[
 M\in\nclass{r+1}\quad\Longrightarrow\quad Q_r(M)\in\nclass1.
\]
For $r=0$ the premise is the conclusion; at $r+1$, first eliminate
degree $r+1$ and then apply the induction hypothesis. By
Lemma~\ref{lem:uniform-ceiling}, this applies to $r=K(t)$ and $M=M_t$.
Each pass preserves context and type, so $N_t$ is closed and ground.
Iterating Lemma~\ref{lem:pass-sound} gives
$\deni{N_t}=\deni{M_t}=D(t)$ by~\eqref{eq:translated-dialogue}.
\qed
\end{proof}

There are $K(t)$ passes, each acting on the entire infinitary term.
A pass may contract infinitely many redexes: the finite bound counts
passes, not elementary reductions. The common ceiling for all branches
comes from the finite source.

\begin{lemma}[Closed ground normal forms]\label{lem:ground-normal-forms}
The closed ground $\beta$-normal terms are exactly the well-founded
terms generated by
\[
 N::=\lit n\mid\query_i(n\mapsto N_n)
       \mid\caseop(N;\,n\mapsto N_n).
\]
\end{lemma}
\begin{proof}
A term in the empty context cannot be a variable, and a ground term
cannot be an abstraction. Suppose a closed normal term is an
application. Follow its function-position subterms until reaching a
root that is not an application; well-foundedness ensures termination.
This head cannot be a variable, by closedness, or an abstraction, since
it would form a redex with its adjacent argument. Nor can it be a
literal, query, or selection, since these have ground type and cannot
occur in function position. This excludes applications. The remaining
constructors have closed ground children inheriting normality, so
structural induction gives the displayed grammar. Conversely, this
grammar introduces no $\beta$-redexes.
\qed
\end{proof}

Selections retain their grafting semantics. Section~\ref{sec:rank}
estimates their contribution to tree height along with that of queries.

\section{Ordinal Bounds}
\label{sec:rank}

We bound the initial rank of $M_t$, its growth under normalisation, and
the dialogue height of $N_t$. We identify natural numbers with finite
ordinals; addition and exponentiation are ordinary ordinal operations,
and $\alpha\ordjoin\beta$ denotes binary supremum.

\subsection{Rank and the Initial Translation}

Define the rank of an auxiliary term by structural recursion:
\begin{equation}\label{eq:term-rank}
 \begin{aligned}
  \rho(x)&=\rho(\lit n)=0,\\
  \rho(\lambda x^\sigma.M)&=\rho(M)+1,\\
  \rho(F\,V)&=(\rho(F)\ordjoin\rho(V))+1,\\
  \rho(\query_i(n\mapsto M_n))&=\bigl(\sup_n\rho(M_n)\bigr)+1,\\
  \rho(\caseop(M;\,n\mapsto M_n))
    &=\bigl(\rho(M)\ordjoin\sup_n\rho(M_n)\bigr)+1.
 \end{aligned}
\end{equation}
At every non-leaf node, the supremum of the immediate child ranks is a
common bound strictly below the rank of the parent. The position of the
successor is essential: child ranks $0,1,2,\ldots$ give parent rank
$\omega+1$, whereas the same sequence of child heights gives dialogue
height $\omega$ in~\eqref{eq:tree-height}.

To bound the starting rank, first count the constructors introduced by
selection at a function type. Write
$\sigma=\sigma_1\to\cdots\to\sigma_k\to\iota$.
Here $k$ counts successive argument positions; arrows inside the
$\sigma_i$ are not counted. The definition of $\caseop_\sigma$ introduces
one ground selection and one abstraction and application per argument.

\begin{lemma}[Selection rank]\label{lem:selection-rank}
For $\sigma=\sigma_1\to\cdots\to\sigma_k\to\iota$, if a selector $M$
and all branches $M_n$ have rank at most $\alpha$, then
\[
 \rho(\caseop_\sigma(M;\,n\mapsto M_n))\le\alpha+(2k+1).
\]
\end{lemma}
\begin{proof}
Induct on $k$ using~\eqref{eq:derived-selection}. At $k=0$, ground
selection has rank at most $\alpha+1$. For $k>0$, weakening preserves
rank and the fresh variable has rank zero, so each new application
$M_n\,x$ has rank at most $\alpha+1$. The induction hypothesis for the
codomain gives $(\alpha+1)+(2k-1)$; the outer abstraction adds one,
giving $\alpha+(2k+1)$ by associativity and finite arithmetic.
\qed
\end{proof}

Define the \emph{source offset} $\mu(u)\in\N$ by recursion on the finite
extended source term $u$:
\begin{equation}\label{eq:source-offset}
 \begin{aligned}
 \mu(x)&=\mu(\Zero)=0,\qquad \mu(\Succ)=2,\qquad \mu(\Omega)=3,\\
 \mu(\Rec_\sigma)&=2k+4
   \quad(\sigma=\sigma_1\to\cdots\to\sigma_k\to\iota),\\
 \mu(\lambda x^\sigma.u)&=\mu(u)+1,\\
 \mu(u\,v)&=\max\{\mu(u),\mu(v)\}+1.
 \end{aligned}
\end{equation}
The recursor clause combines the selection overhead $2k+1$ with the
three outer abstractions of $R^\infty_\sigma$. For the closed functional
$t$, define
\begin{equation}\label{eq:initial-offset}
 m(t)=\mu(t\,\Omega).
\end{equation}
Thus $K(t)$ bounds the number of passes, whereas $m(t)$ bounds the
finite offset of the initial rank.

\begin{lemma}[Initial rank]\label{lem:initial-rank}
For every finite extended source term $u$,
$\rho(\widehat u)\le\omega+\mu(u)$. In particular,
\[
 \rho(M_t)\le\omega+m(t).
\]
\end{lemma}
\begin{proof}
For the recursor clause, write $\sigma=\sigma_1\to\cdots\to\sigma_k\to\iota$.
Inside its template, $f,a,z$ are variables of rank zero.
Induction on $n$ gives $\rho(I_n(f,a))\le2n$: the base is zero, and
\[
 \rho(I_{n+1}(f,a))
  =(1\ordjoin\rho(I_n(f,a)))+1
  \le\max\{1,2n\}+1\le2(n+1).
\]
The branch ranks therefore have common bound $\omega$.
Lemma~\ref{lem:selection-rank} and the three outer abstractions give
\begin{equation}\label{eq:recursor-initial-rank}
 \rho(R^\infty_\sigma)\le\omega+(2k+4).
\end{equation}
The successor and oracle templates have ranks at most $2$ and $3$;
variables and zero have rank zero. These are the constant and variable
cases of~\eqref{eq:source-offset}.
For abstraction, taking successor adds one to the finite offset.
For application, use
$(\omega+a)\ordjoin(\omega+b)=\omega+\max\{a,b\}$, then take successor.
Structural induction on $u$ proves the bound with exactly the offset
in~\eqref{eq:source-offset};~\eqref{eq:initial-offset} gives the case $M_t$.
\qed
\end{proof}

\subsection{Rank Growth under Normalisation}

\begin{lemma}[Substitution rank]\label{lem:substitution-rank}
Renaming preserves rank. If every image of a typed simultaneous
substitution $\vartheta$ has rank at most $\beta$, then
\[
 \rho(M[\vartheta])\le\beta+\rho(M).
\]
In particular, $\rho(U[V/x])\le\rho(V)+\rho(U)$.
\end{lemma}
\begin{proof}
Renaming preserves rank by structural induction, since it preserves
constructors and the ranks of all children. For substitution, induct
on $M$. Variables use the assumed bound and $\beta+0=\beta$; literals
have rank zero. Under a binder, the lifted substitution fixes the new
variable, of rank zero, and weakens the old images without changing
their ranks. The same bound $\beta$ therefore applies, and the
abstraction case follows by induction and associativity of addition.

For an application $F\,V$, put $S=\rho(F)\ordjoin\rho(V)$.
Both substituted children have rank at most $\beta+S$, whence
\[
 \rho((F\,V)[\vartheta])
 \le(\beta+S)+1=\beta+(S+1)=\beta+\rho(F\,V).
\]
For a query, put $S=\sup_n\rho(M_n)$. The induction hypotheses give
$\rho(M_n[\vartheta])\le\beta+\rho(M_n)\le\beta+S$ for every $n$.
Taking the supremum and then successor gives the same bound
$\beta+(S+1)$. For a selection, use
$S=\rho(M)\ordjoin\sup_n\rho(M_n)$ to bound both the substituted
selector and every substituted branch before taking successor.
Finally, single substitution has common image bound $\rho(V)$, since
the other variables are left fixed and have rank zero.
\qed
\end{proof}

Following Martin-L\"of's quantitative analysis~\cite[\S\S4--5]{martinlof1972},
we use, for every ordinal $\alpha$,
\begin{equation}\label{eq:exponential-room}
 2^\alpha+1\le 2^\alpha+2^\alpha=2^{\alpha+1}.
\end{equation}
The equality is the successor law for ordinal exponentiation;
the inequality uses $1\le2^\alpha$.

\begin{lemma}[Exponential growth bound]\label{lem:pass-rank}
For every $r\in\N$ and every well-founded infinitary term $M$,
\[
 \rho(\pass r M)\le2^{\rho(M)}.
\]
\end{lemma}
\begin{proof}
Induct on $M$. Variables and literals remain of rank zero, below
$2^0=1$. At a non-leaf node, let $S$ be the supremum of the
immediate child ranks, so that $\rho(M)=S+1$. By induction and
monotonicity of exponentiation, all processed children have rank at
most $2^S$.

If the pass rebuilds the node, its new rank is at most
$2^S+1\le2^{S+1}$. This includes the countable branching
cases, since $2^S$ bounds the entire family of processed children.
If an application is contracted, its processed function is
$\lambda x.U'$ and its processed argument is $V'$, both of rank at
most $2^S$. Since
$\rho(U')\le\rho(\lambda x.U')$, Lemma~\ref{lem:substitution-rank}
and~\eqref{eq:exponential-room} give
\[
 \rho(U'[V'/x])\le\rho(V')+\rho(U')
 \le2^S+2^S=2^{S+1}.
\]
The bound on $\rho(U')$ follows from that on the processed abstraction
$\lambda x.U'$, including when processing the function has exposed it.
\qed
\end{proof}

The degree argument in Section~\ref{sec:normalisation} determines how many
passes suffice. The rank estimate applies to every auxiliary term,
without a bound on its redex degrees.

\subsection{From Normal Forms to Tree Height}

Ground selection requires a bound on the effect of grafting.

\begin{lemma}[Grafting height]\label{lem:grafting-height}
Let $d\in\Tree\N$, $g:\N\to\Tree\N$, and let $b$ be an ordinal.
If $h(g(n))\le b$ for every $n$, then
\[
 h(d\bind g)\le b+h(d).
\]
\end{lemma}
\begin{proof}
Induct on $d$. At a leaf $\eta(n)$ the claim is the assumed bound on
$g(n)$, since $b+0=b$. At $d=\beta_i(\varphi)$, the induction
hypothesis gives, for each $n$,
\[
 \begin{aligned}
  h(\varphi(n)\bind g)+1
   &\le(b+h(\varphi(n)))+1\\
   &=b+(h(\varphi(n))+1)\le b+h(d).
 \end{aligned}
\]
Taking the supremum over $n$ proves the node case by
\eqref{eq:tree-height} and~\eqref{eq:tree-bind}.
\qed
\end{proof}

\begin{lemma}[Ground height]\label{lem:ground-height}
If $N$ is a closed ground $\beta$-normal term, then
\[
 h(\deni N)\le2^{\rho(N)}.
\]
\end{lemma}
\begin{proof}
Induct on the grammar of Lemma~\ref{lem:ground-normal-forms}.
A literal has height zero and rank zero. For
$N=\query_i(n\mapsto N_n)$, put $S=\sup_n\rho(N_n)$.
The induction hypotheses bound every child height by $2^S$, so
\[
 h(\deni N)\le2^S+1\le2^{S+1}=2^{\rho(N)}.
\]
For $N=\caseop(M;\,n\mapsto N_n)$, put
$S=\rho(M)\ordjoin\sup_n\rho(N_n)$. The selector and every branch
have denotations of height at most $2^S$. The semantics
\eqref{eq:infinitary-semantics} and Lemma~\ref{lem:grafting-height}
give
\[
 h(\deni N)\le2^S+h(\deni M)
 \le2^S+2^S=2^{S+1}=2^{\rho(N)}.
\]
Both node cases use~\eqref{eq:exponential-room}.
\qed
\end{proof}

\subsection{The Height Bound}

Write $E_2(\alpha)=2^\alpha$ and let its superscript denote function
iteration:
\[
 E_2^0(\alpha)=\alpha,\qquad
 E_2^{n+1}(\alpha)=2^{E_2^n(\alpha)}.
\]

\begin{proof}[of Theorem~\ref{thm:main}]
Put $m=m(t)$ as defined in~\eqref{eq:initial-offset}.
Iterating Lemma~\ref{lem:pass-rank} through the $K(t)$ passes
in~\eqref{eq:passes} gives
\[
 \rho(N_t)\le E_2^{K(t)}(\rho(M_t))
             \le E_2^{K(t)}(\omega+m).
\]
Theorem~\ref{thm:normalisation} and Lemma~\ref{lem:ground-height}
therefore give the more precise bound
\begin{equation}\label{eq:fine-height-bound}
 h(D(t))=h(\deni{N_t})\le2^{\rho(N_t)}
                    \le E_2^{K(t)+1}(\omega+m).
\end{equation}

Since $2^\omega=\sup_{n<\omega}2^n=\omega$, the exponent laws
$a^{\beta+\gamma}=a^\beta\cdot a^\gamma$ and
$a^{\beta\cdot\gamma}=(a^\beta)^\gamma$ yield
\[
 \begin{aligned}
  E_2(\omega+m)&=\omega\cdot2^m,\\
  E_2^2(\omega+m)&=\omega^{2^m}<\omega^\omega=\tow1,
 \end{aligned}
\]
where $2^m$ is finite. Induction on $n\ge2$ now gives
$E_2^n(\omega+m)<\tow{n-1}$: at the successor step,
monotonicity in the base and strict monotonicity of $\omega$
exponentiation give
\[
 E_2^{n+1}(\omega+m)
 \le\omega^{E_2^n(\omega+m)}
 <\omega^{\tow{n-1}}=\tow n.
\]
Since $K(t)\ge2$, the inductive bound at $n=K(t)+1$ together
with~\eqref{eq:fine-height-bound} gives
\begin{equation}\label{eq:final-height-bound}
 h(D(t))<\tow{K(t)}.
\end{equation}
Finally, $\tow0<\tow1$, and strict monotonicity of exponentiation
gives $\tow n<\tow{n+1}$ by induction. Hence
$\tow n<\tow{n+1}\le\sup_j\tow j=\eps$ for every $n$, completing
the bound.
\qed
\end{proof}

The tower supremum is the least fixed point of $\alpha\mapsto\omega^\alpha$.
Indeed, continuity gives
\[
 \omega^{\eps}=\sup_n\omega^{\tow n}=\sup_n\tow{n+1}=\eps.
\]
If $\omega^\alpha=\alpha$, positivity gives $1\le\alpha$, hence
$\omega\le\alpha$. Induction bounds every $\tow n$ by $\alpha$,
so $\eps\le\alpha$.

For $t_\omega$ from Example~\ref{ex:omega-height},
\eqref{eq:source-ceiling}--\eqref{eq:source-K} give $K(t_\omega)=2$.
The general bound is therefore $h(D(t_\omega))<\tow2$;
the direct calculation gives its exact height $\omega$.
The parameter $K(t)$ bounds type levels throughout the source syntax.
We do not claim that the bound is sharp.

\subsection{The Combinatory Interpretation}
\label{sec:combinatory}

The original dialogue interpretation~\cite{escardo2013} uses combinatory
System~T: terms are built by application from zero, successor, typed
combinators $\mathsf K,\mathsf S$, and an iterator
$\mathsf{Iter}_\sigma:(\sigma\to\sigma)\to\sigma\to\iota\to\sigma$.
Write $D_{\mathrm c}(s)$ for its specified tree and $E(s)$ for its
finite $\lambda$-encoding. The encoding identifies the simple-type
grammars, preserves zero, successor, and application, and sets
\[
 \begin{aligned}
 E(\mathsf K)&=\lambda x\,y.x,\qquad
 E(\mathsf S)=\lambda f\,g\,x.f\,x\,(g\,x),\\
 E(\mathsf{Iter}_\sigma)
   &=\lambda f^{\sigma\to\sigma}.\Rec_\sigma(\lambda i^\iota.f).
 \end{aligned}
\]
\begin{corollary}[Original combinatory interpretation]
\label{cor:combinatory-bound}
For every closed combinatory System~T term
$s:(\iota\to\iota)\to\iota$,
\[
 D_{\mathrm c}(s)=D(E(s)),
 \qquad
 h(D_{\mathrm c}(s))<\tow{K(E(s))}<\eps.
\]
\end{corollary}
\begin{proof}
Write $W_\sigma$ for the combinatory dialogue value spaces, defined as
in~\eqref{eq:dialogue-values}, and $\llbracket s\rrbracket^{\mathrm c}\in W_\sigma$
for their interpretation. Define $R_\sigma\subseteq V_\sigma\times W_\sigma$ by
\[
 \begin{aligned}
 R_\iota(d,e)&\iff d=e,\\
 R_{\sigma\to\tau}(f,g)&\iff
  \forall a\in V_\sigma\,\forall b\in W_\sigma.\,
  R_\sigma(a,b)\Rightarrow R_\tau(f(a),g(b)).
 \end{aligned}
\]
Induction on $\sigma$ shows that pointwise related branch families
remain related after typed grafting onto the same ground tree.
For related functions and initial values, induction on $n$ relates their
$n$-fold iterates. The recursor step in $E(\mathsf{Iter}_\sigma)$ ignores
its index, so this gives the iterator case after grafting.
The other constants satisfy the relation directly, and application uses
its arrow clause. Structural induction therefore gives
$R_\sigma(\dend{E(s)},\llbracket s\rrbracket^{\mathrm c})$.
The generic inputs are related because they preserve tree equality.
Applying the arrow clause yields $D(E(s))=D_{\mathrm c}(s)$;
Theorem~\ref{thm:main} gives the height bound.
\qed
\end{proof}

\section{Mechanisation}
\label{sec:mechanisation}

We formalise the constructions and proofs of
Sections~\ref{sec:translation}--\ref{sec:rank} in Agda~\cite{norell2007},
using TypeTopology~\cite{typetopology}. The syntax is indexed by context and
type, with de Bruijn variables and typed renaming and substitution.
An explicit map between the library's two simple-type datatypes
implements the adapter of Corollary~\ref{cor:combinatory-bound}.

Following univalent foundations~\cite{hott2013}, ordinals have carriers
in $\mathcal{U}_0$. The ordinal modules assume univalence, propositional
truncation, set replacement, and excluded middle for propositions in
$\mathcal{U}_1$. Syntactic substitution and degree elimination require
no classical assumptions. The semantic modules take function
extensionality, derived from univalence in the final assembly.
All local modules use \texttt{-{}-safe} and \texttt{-{}-without-K},
with assumptions as explicit module parameters.
\ifdefined\FossacsPreprint
\else
The Data Availability
Statement identifies the accompanying code and reproduction instructions.
\fi

\section{Conclusion}
\label{sec:conclusion}

We have proved $h(D(t))<\tow{K(t)}<\eps$ with $K(t)$ computed from the
finite System~T source. Exact tree preservation connects finite-pass
normalisation to dialogue height. The Agda formalisation covers both
the $\lambda$-calculus and original combinatory interpretations.

It remains to determine whether these heights are cofinal in $\eps$
and to develop a constructive ordinal-code version of the analysis.
Connections between constructive ordinal representations~\cite{kraus2021}
offer a starting point for the latter.

\clearpage
\ifdefined\FossacsPreprint
\else
\paragraph{Data Availability Statement.}
The Agda formalisation and reproduction instructions are provided in the
anonymised supplementary archive \texttt{dialogue-height-artifact.zip}
accompanying this manuscript.
The README specifies the Agda version, the pinned TypeTopology revision,
the checking commands, and the correspondence between the paper's results
and their formal statements.
\fi

\bibliographystyle{splncs04}
\bibliography{references}
\end{document}